\documentclass[11pt]{amsart}

\usepackage{mathpple}
\usepackage{amsmath,amsfonts,amssymb}
\usepackage{amsrefs,mathrsfs}
\usepackage{amsthm,color}
\usepackage[all]{xy}
\usepackage{stmaryrd}
\usepackage{multirow}

\usepackage{hyperref,graphicx}
\usepackage{tikz}
\usetikzlibrary{arrows,backgrounds,decorations.markings,matrix}
\numberwithin{equation}{section}

\usepackage{MnSymbol}

\newcommand{\C}{\mathbb{C}}

\newcommand{\R}{\mathbb{R}}

\newcommand{\Z}{\mathbb{Z}}
\renewcommand{\H}{\mathbb{H}}

\newcommand{\E}{{\mathcal E}}
\newcommand{\F}{\mathbf{F}}

\newcommand{\g}{\mathbf{g}}

\newcommand{\J}{\mathscr{J}}

\newcommand{\abracket}[1]{\left\langle#1\right\rangle}
\newcommand{\bbracket}[1]{\left[#1\right]}
\newcommand{\fbracket}[1]{\left\{#1\right\}}
\newcommand{\bracket}[1]{\left(#1\right)}

\newcommand{\aabracket}[1]{\left\llangle#1\right\rrangle}

\newcommand{\llb}{\llbracket}
\newcommand{\rrb}{\rrbracket}
\newcommand{\llp}{(\!(}
\newcommand{\rrp}{)\!)}

\newcommand{\mc}{\mathcal}
\newcommand{\mr}{\mathrm}

\newcommand{\pa}{\partial}

\renewcommand{\dbar}{\bar\pa}
\newcommand{\OO}{{\mathcal O}}

\newcommand{\BV}{Batalin--Vilkovisky }

\newcommand{\into}{\hookrightarrow}
\newcommand{\Ol}{\mathcal O_{\mr{loc}}}

\newcommand{\iso}{\cong}

\newcommand{\bb}{\mathbf{b}}

\DeclareMathOperator{\HH}{H}

\DeclareMathOperator{\Aut}{Aut}

\DeclareMathOperator{\Sym}{Sym}
\DeclareMathOperator{\Hom}{Hom}

\DeclareMathOperator{\Tr}{Tr}

\DeclareMathOperator{\PV}{PV}

\DeclareMathOperator{\Obs}{Obs}

\DeclareMathOperator{\Dens}{Dens}

\newcommand{\A}{\mathcal A}

\renewcommand{\L}{\mathcal L}

\theoremstyle{plain}
\newtheorem{thm}{Theorem}[section]
\newtheorem{thm-defn}{Theorem/Definition}[section]
\newtheorem{lem-defn}[thm]{Lemma/Definition}

\newtheorem{prop}[thm]{Proposition}

\newtheorem{cor}[thm]{Corollary}

\theoremstyle{definition}
\newtheorem{defn}[thm]{Definition}
\newtheorem{eg}[thm]{Example}

\theoremstyle{remark}
\newtheorem{rmk}[thm]{Remark}

\allowdisplaybreaks[4]  

	\def\curved{\tikz[baseline=.1ex]{
			\draw[ ->] (0,0.32) arc (30:330:0.5);}
	}

\begin{document}

  \title{\mbox{Dispersionless Integrable Hierarchy via Kodaira--Spencer Gravity}}
 \author{Weiqiang He, Si Li, Xinxing Tang, and Philsang Yoo}

  \address{
W. He: Department of Mathematics, Sun Yat-Sen University, Guangdong, China;
}
\email{hewq@mail2.sysu.edu.cn}
  \address{
S. Li: 
\begin{tabular}{l}
Yau Mathematical Sciences Center, Tsinghua University, Beijing, China\\ Institute for Advanced Study, New Jersey, USA
\end{tabular}
}

\email{sili@mail.tsinghua.edu.cn}
  \address{
X. Tang: Yau Mathematical Sciences Center, Tsinghua University, Beijing, China;
}
\email{tangxinxing@mail.tsinghua.edu.cn}
  \address{
P. Yoo: Department of Mathematics, Yale University, New Haven, USA;
}
\email{philsang.yoo@yale.edu}

\date{}
\maketitle

\begin{abstract} We explain how dispersionless integrable hierarchy in 2d topological field theory arises from the Kodaira--Spencer gravity (BCOV theory). The infinitely many commuting Hamiltonians  are given by the current observables associated to the infinite abelian symmetries of the Kodaira--Spencer gravity. We describe a BV framework of effective field theories that leads to the B-model interpretation of dispersionless integrable hierarchy.
\end{abstract}

\tableofcontents

\section{Introduction}

It is extensively studied in the literature that 2d topological field theories give rise to integrable hierarchies. In this paper, we present a natural connection between the classical dispersionless integrable hierarchy in 2d topological field theory and the classical geometry of Kodaira--Spencer gravity (or BCOV theory \cite{BCOV}). The integrable hierarchy structure arises when we extend BCOV's theory \cite{BCOV} to include gravitational descendants as formulated in \cite{Si-Kevin}. This is motivated by studying the projection (as proposed in \cite{Si-review})
 $$
 \pi\colon  X\times \Sigma \to \Sigma.
 $$
 Here $\Sigma$ is a surface with a Calabi--Yau structure, namely, the complex plane $\C$, the punctured plane $\C^\times$, or an elliptic curve $E$; $X$ is the target space of 2d topological field theory; and $\pi$ is the natural projection map. 

We put B-model on $\Sigma$ (this is where we need Calabi--Yau structure), and imagine putting either A-model or B-model on $X$. Compactifying along $X$, we obtain an effective 2d chiral field theory on $\Sigma$ that encodes geometric invariants of $X$.  This field theory can be viewed as a generalized Kodaira--Spencer gravity in a similar fashion as formulated in \cite{Si-Kevin,vertex}. We find that the induced BV master equation of such effective theory encodes precisely the integrable hierarchy structure associated to the target $X$. Geometrically, the infinitely many mutually commuting Hamiltonians originate from the current observables associated to the infinite abelian symmetries of the effective 2d chiral theory.

Motivated by this, we introduce the notion of $H$-valued BCOV theory for a small phase space $H$ of a 2-dimensional topological field theory. The effective field theory as above will be of this type. We explain the corresponding classical dispersionless integrable hierarchy in this generality. Note that such interpretation of integrable hierarchy was used in \cite{L-elliptic,vertex} even at the quantum level to  solve the full higher genus B-model on elliptic curves generalizing \cite{Dijkgraaf-elliptic} (this corresponds to the special case when $X=\mr{pt}, \Sigma=$ elliptic curve). It would be extremely interesting to extend our classical analysis to the quantum case in general, i.e., the integrability meaning of the quantum master equation for the resulting effective theory on $\Sigma$. We hope to address this issue in some future work.

The description of the effective field theory on $\Sigma$ requires an essential use of BV formalism \cite{BV} (in fact, a degenerate BV theory). It leads to a generalization of BCOV theory as formulated in \cite{Si-Kevin}. In the first part of this paper, we present a review on BV formalism that will be used extensively in this paper. We also introduce a BV framework for establishing the effective field theory obtained by compactifying an internal space.  The second part of this paper is mainly to explain the precise connection between dispersionless integrable hierarchy and the classical BV master equation of BCOV type theory.

 \noindent \textbf{Conventions}:
 \begin{itemize}
\item Let $V=\bigoplus\limits_k V_k$ be a graded vector space, with $V_k$ being the degree $k$ components. We denote $V[m]$ by the degree $m$ shifted space such that $V[m]_k:= V_{m+k}$.
\item Let $(\A,\cdot)$ be a graded commutative algebra. We denote the graded commutator by $[-,-]$, i.e., if we write the degree of elements $a,b \in \A$ as $|a|, |b|$, then 
$$
   [a,b]:= a\cdot b-(-1)^{|a| |b|}b \cdot a.
$$
In general, we always use the Koszul sign rule in dealing with graded objects.
\end{itemize}

 \noindent \textbf{Acknowledgments}. The work of S.L. is partially supported by grant 11801300 of NSFC  and grant Z180003 of Beijing Natural Science Foundation. The work of X.T. is partially supported by Tsinghua Postdoc Grant 100410058. Part of this work was done while S.L. was visiting Institute for Advanced Study in Fall 2019. S.L. thanks for their hospitality and provision of excellent working enviroment.

 \section{Classical BV formalism}

We collect geometric basics on field theories in the \BV (BV) formalism to be used in this paper. For a slightly different presentation with more details on the subject related to the current paper, one may want to refer to the books  \cite{CostelloBook,CG2}.

\subsection{Free BV theory}\label{sec:free-BV}
We will define a classical BV theory as a free part  together with additional specified data encoding the interaction.  Let us start with the free BV theory which, roughly speaking, is encoded as a sheaf on a spacetime manifold with $(-1)$-shifted symplectic structure (or a $(-1)$-shifted Poisson structure in the degenerate case).

 \begin{defn}
 A \emph{free classical BV theory} on a manifold $M$ consists of $(E, Q,\langle -,-\rangle)$ where
 \begin{itemize}
 \item[(1)] $(E,Q)$ is a complex of vector bundles on $M$, where $Q$ is a differential operator making $(\E,Q)$ an elliptic complex. Here $\E=\Gamma(M, E)$ are smooth sections of $E$.
 \item[(2)] a bundle morphism of degree $-1$
 $$
     \abracket{-,-}\colon E\otimes E \to \Dens(M)
  $$
  which is a fiberwise skew-symmetric and non-degenerate pairing. Here $\Dens(M)$ is the density line bundle at cohomology degree $0$. It defines
  $$
      \omega(\alpha, \beta):=\int_M \abracket{\alpha, \beta}, \quad \alpha, \beta\in \E_{c}
  $$
  which can be viewed as a symplectic pairing of degree $-1$. Here $\E_{c}\subset \E$ is the subspace of sections with compact support.
  \item[(3)]  the differential operator $Q$ is graded skew self-adjoint with respect to $\omega$:
  $$
    \omega(Q\alpha, \beta)=-(-1)^{|\alpha|}\omega(\alpha, Q\beta), \quad \forall \alpha, \beta\in \E_{c}.
  $$
  Here $|\alpha|$ is the degree of $\alpha$.
 \end{itemize}
\end{defn}

As the last condition amounts to requiring $\omega$ to be $Q$-closed, we say the triple $(\E, Q, \omega)$ defines a $(-1)$-shifted dg symplectic space. We often abuse a notation to denote a free classical BV theory by $(\E, Q, \omega)$. We will always assume $M$ is oriented and hence identify the density bundle $\Dens(M)$ with the bundle of top differential forms on $M$.

\begin{defn}
We define the space of \emph{functionals} on $\E$ to be
 \[
 \OO(\E): =  \prod_{k\geq 0  } \Hom( \E^{\otimes k} , \C  )_{S_k}.
 \]
 Here $\E^{\otimes k}$ means the following completed tensor product (smooth sections of $E^{\boxtimes k}$)
 $$
 \E^{\otimes k}:=\Gamma(\underbrace{X\times \cdots \times X}_{k}, E\boxtimes \cdots \boxtimes E). 
 $$
$\Hom$ means continuous maps (i.e. distributions), and $(-)_{S_k}$ means taking $S_k$-coinvariants with respect to the graded permutation. Elements of $\Hom( \E^{\otimes k} , \C  )_{S_k}$ can be viewed as $k$-th degree homogenous polynomial functions on the graded vector space $\E$.
 \end{defn}

If $M$ is compact, then a functional $F \in \OO(\E)$ is called \emph{local} if each Taylor component $F_k \in \Hom( \E^{\otimes k} , \C   )_{S_k}$ is a finite sum of terms of the form \[ F_k (\phi) = \int_M (D_1 \phi)\cdots (D_k \phi)  d\mu_M, \] where $D_i$ is a differential operator and $d\mu_M$ is a volume form of $M$. If $M$ is noncompact, we simply have the same expression but viewing it as a functional well-defined on $\E_c$. 

 Alternately, we can define the space of local functionals by
 $$
 \Ol(\E) = \Dens(M) \otimes_{\mc D_M} \mc O(J(E))$$
 where $\mc D_M$ is the sheaf of differential operators of $M$ and $J(E)$ is the sheaf of sections of the $\infty$-jet bundle of $E$; a section of $\Ol(\E)$ should be thought of as for input $\phi\in\mc E$ producing a density element $F(\phi)$ on $M$ such that $F(\phi)(x)$ depends only on the $\infty$-jet of $\phi$ at $x$, modulo total derivatives. This captures the intuition that physics must be local and have no spooky action at a distance.

To proceed, one has to find a Poisson structure on a space of functionals. In the finite-dimensional setting, the Poisson kernel is simply the inverse matrix of the symplectic pairing.  In our infinite-dimensional setting, this amounts to asking the Poisson kernel to be the integral kernel of the identity operator with respect to $\omega$, that is, $K_0$ satisfying
 \[\omega(K_0 (x,y),\phi(y) )=\phi(x)\qquad\text{for all}\quad \phi\in \mc E.\]
  Since $\omega$ is given by integration, $K_0$ behaves like a $\delta$-function. We can view $K_0$ as a distributional section of $E\times E$ supported on the diagonal of $M\times M$. It can be checked that $K_0$ is graded symmetric, and hence defines a distributional section of $\Sym^2(\E)$.

It is worth pointing out that the Poisson kernel $K_0$ is in fact graded symmetric instead of skew-symmetric as in symplectic geometry. This is due to the fact that our symplectic pairing has odd degree. Let us explain this graded symmetry in the case when $M$ is 0-dimensional and hence $\mc E=V$ is a finite-dimensional graded vector space. Let $\omega$ be a $(-1)$-symplectic pairing on $V$. Let $V^*$ be the linear dual. Then $\omega$ defines an isomorphism
$$
\omega: V^*\to V[1]
$$
which induces an identification
\[\bigwedge {}\! ^2 (V^*)\to \bigwedge{}\! ^2(V[1])\iso \Sym^2(V)[2]. \]
Here we have used the canonical identification $\bigwedge^k(V[1])\simeq \Sym^k(V)[k]$ for any graded vector space $V$. Under this identification, the symplectic form $\omega$, which is viewed as a degree $-1$ element of $\bigwedge^2 (V^*)$, is sent to $K_0$, a degree $1$ element of $\Sym^2(V)$. Then $K_0$ is precisely the inverse Poisson kernel of $\omega$. The same reasoning applies to our infinite-dimensional space $\E$. We leave this check to careful readers.

\begin{defn}
The \emph{BV kernel} of a free classical BV theory $(E, Q, \langle -,-\rangle)$ is defined to be the degree 1 distributional section $K_0$ of $\Sym^2(\E)$  such that
$$
\omega(K_0 (x,y),\phi(y) )=\phi(x), \quad \forall \phi \in \E.
$$
\end{defn}

\begin{rmk}The singularity of $K_0$ around the diagonal signals the infinite-dimensional nature and is at the heart of the ultra-violet divergence problem in quantum field theory. \end{rmk}

With the Poisson kernel $K_0$, we would like to define a Poisson bracket on a space of functionals as in  the finite-dimensional case. However, because of the distributional nature of $K_0$, this does not work on $\OO(\E)$. Instead, it makes sense on $\Ol(\E)$, defining a degree 1 bracket (called the \emph{BV bracket})
\[\{-,-\}_{\mr{BV}}\colon \Ol(\E)\times \Ol(\E)\to \Ol(\E).\]

This can be seen as follows.  Let $S\in  \Ol(\E)$ be a local functional. If we write a generic field by  $\phi\in\E$, then the variation of $S$ can be uniquely brought into the form
$$
\delta S= \int_M \abracket{\delta \phi, \mc L_S (\phi) }
$$
by removing all derivatives on $\delta \phi$ via integration by parts. It is easy to see that $\mc L_S(\phi)$ is a uniquely determimed local expression that consists of derivatives of fields. Then the BV bracket between local functionals can be defined by
$$
  \fbracket{S_1, S_2}_{\mr{BV}}= \int_M \abracket{\mc L_{S_1}, \mc L_{S_2}}.
$$
Intuitively, this is precisely how the Poisson bracket is defined in the finite-dimensional case. The local functional $S$ defines a Hamiltonian vector field $\{S,-\}_{\mr{BV}}$, which can be identified as an infinitesimal transformation on fields by
$$
\delta_S \phi =\L_S(\phi).
$$
Then the BV bracket between two local functionals $S_1, S_2$ can be written as
$$
\{S_1, S_2\}_{\mr{BV}}=\delta_{S_1}S_2= \int_M \abracket{\mc L_{S_1}, \mc L_{S_2}}.
$$

The differential $Q$ on $\E$ induces dually a differential on functional spaces. Explicitly, let us consider 
$S \in \Hom( \E^{\otimes k} , \C)$ which can be viewed as a multi-linear map 
$$
S:  \phi_1\otimes\cdots \otimes \phi_k \to S(\phi_1, \cdots, \phi_k), \quad \phi_i \in \E.
$$
Then $QS \in \Hom( \E^{\otimes k} , \C)$ is defined by 
$$ 
QS(\phi_1, \cdots, \phi_k):=\sum_{i=1}^k\pm S(\phi_1, \cdots, Q\phi_i, \cdots, \phi_k). 
$$
Here $\pm$ is the Koszul sign by permuting graded objects. Since $Q$ is a differential operator, it preserves locality and defines a differential on  $\Ol(\E)$.  The fact $Q(\omega)=0$ implies that $Q$ is compatible with the BV bracket. This defines a dg Lie algebra
$$
(\Ol(\E)[-1], Q, \{ -,-\}_{\mr{BV}})
$$
on the degree $-1$ shifted space $\Ol(\E)[-1]$. We illustrate this by the Chern--Simons theory below in Example \ref{eg:CS}. For more details in generality we refer to \cite[Section 5.3]{CostelloBook}.

Taking the BV bracket with a local functional can be defined on a larger space, called observables. Briefly speaking, observables are defined to be functions on fields, which we will  denote by $\Obs$. A theory of observables in BV formalism is systematically developed in \cite{CG2} and we refer there for more precise discussions. For our application later, we will only need the following fact about the BV bracket with local functionals.

Let $J\in \Ol(\E)$ be a local functional. Its Hamiltonian vector field $\{J,-\}_{\mr{BV}}$ induces an infinitesimal transformation $
\delta_J
$. It defines a derivation on observables
$$
\delta_J\colon  \Obs\to \Obs.
$$
In other words, we have a well-defined BV bracket
$$
\{-,-\}_{\mr{BV}}\colon \Ol(\E)\times \Obs \to \Obs, \quad  \{J, O\}_{\mr{BV}}:=\delta_J O.
$$
Note that local functionals can be viewed as observables $\Ol(\E)\subset \Obs$ when the spacetime manifold $M$ is compact. Therefore this above BV bracket extends that on local functionals in that case. We explain some examples of observables and BV brackets in Example \ref{eg:CS}. In Section \ref{sec:3}, we will encounter observables which are supported on codimension 1 subspaces that will play an important role in our description of integrable hierarchy.

\subsection{Classical master equation}

To describe dynamics of a BV theory in general, one has to consider an interaction term. It can be regarded as a deformation of the free theory by a solution of the Maurer--Cartan equation associated to the dg Lie algebra
$
(\Ol(\E)[-1], Q, \{ -,-\}_{\mr{BV}}).
$

 \begin{defn} A local functional $I\in \Ol(\E)$ of degree $0$ is said to satisfy the \emph{classical master equation} if
\[QI+{1\over 2}\{I,I\}_{\mr{BV}}=0.\]
 \end{defn}

 It can be seen that such $I$ gives a new differential
 $$
 Q+\{I,-\}_{\mr{BV}}=Q+\delta_I
 $$
 which squares to zero.  In field theory, $Q+\{I,-\}_{\mr{BV}}$ defines the classical BRST operator, while $Q$ is the leading linearized transformation.

The action functional $S$ associated to the free classical BV theory $(E,Q,\langle -,-\rangle )$ together with a local functional $I$ satisfying the classical master equation is \[S(\phi)  ={1\over 2} \int_M \langle \phi,Q \phi\rangle  + I(\phi), \quad \phi \in \E_{c}. \]
The classical master equation for $I$ is equivalent to the following  traditional form
$$
\{S, S\}_{\mr{BV}}=0.
$$

\begin{eg}[Chern--Simons theory]\label{eg:CS} Let $M$ be a three-dimensional oriented manifold. Let $\g$ be a Lie algebra equipped with a non-degenerate trace pairing $\Tr$. For simplicity, we consider Chern--Simons theory for a trivial principal bundle. The space of fields in the BV formalism is
$$
    \E:= \Omega^\bullet(M)\otimes \g [1].
$$
Here the degree shift ensures that the connection 1-form $\Omega^1(M)\otimes \g$ sits at degree $0$. The $(-1)$-shifted symplectic pairing is given by
$$
\omega(\alpha, \beta)=\int_M \Tr(\alpha\wedge \beta), \quad \alpha\in \Omega_c^{k}\otimes \g, \beta\in \Omega_c^{3-k}\otimes \g.
$$
The differential $Q=d$ is the de Rham differential. The Chern--Simons functional in the BV formalism is given by
$$
CS[\A]=\int {1\over 2}\Tr(\A\wedge d\A)+{1\over 6}\Tr(\A \wedge [\A, \A]), \quad \A\in \E_c.
$$
The first term is the free part, and the second term is the interaction part, denoted by $I$. The BV bracket $\{CS,-\}_{\mr{BV}}$ introduces an infinitesimal transformation $\delta_{CS}$. To compute $\delta_{CS}$, we follow the recipe described above and consider the variation. We find
$$
\delta CS= \int \Tr \bracket{\delta \A\wedge \bracket{d\A+{1\over 2}[\A, \A]}}.
$$
This says that $\L_{CS}(\A)=d\A+{1\over 2}[\A, \A]$, hence
$$
    \delta_{CS}\A=d\A+{1\over 2}[\A, \A].
$$
This formula is read in components as follows. Let us write 
$$
\A:= \sum\limits_{i=0}^3 \A^i, \quad  \text{where}\quad \A^i\in \Omega^i(X)\otimes \g[1].
$$ 
Then $\delta_{CS}\A^i=(\delta_{CS}\A)^i$ equals the $i$-form part of $d\A+{1\over 2}[\A, \A]$. For example, 
$$
  \delta_{CS}\A^0={1\over 2}[\A^0, \A^0], \quad \delta_{CS}\A^1=d\A^0+[\A^0, \A^1]. 
$$
Note that $\delta_{CS}$ is precisely the Chevalley--Eilenberg differential associated to the dg Lie algebra $\Omega^\bullet(X)\otimes \g$. This immediately implies
$$
\delta_{CS}^2=0, \quad \text{or equivalently}\quad \{CS, CS\}_{\mr{BV}}=0.
$$
In terms of the interaction term $I$,  it induces an infinitesimal transformation
$$
\delta_I \A={1\over 2}[\A, \A]
$$
and satisfies the following classical master equation
$$
  d I +{1\over 2}\{I, I\}_{\mr{BV}}=0.
$$

Next we explain some examples of observables. 

Given $\xi \in \g^*$ and a point $x\in M$, we define a linear observable $\mc O_x^{\xi}$ of degree 1 by 
$$
\mc O_x^{\xi} : \E\to \C, \quad \A\to  \mc O_x^{\xi}(\A):=\langle \xi, \mc A^0(x)\rangle. 
$$
In general, given $
\Xi= \xi_1\wedge \cdots \wedge \xi_k \in \wedge^k \g^*
$, we define a degree $k$ observable $\mc O_x^{\Xi}$ by 
$$
\mc O_x^{\Xi}(\A):= \mc O_x^{\xi_1}(\A) \cdots  \mc O_x^{\xi_k}(\A).
$$
This defines a map 
$$
\mc O_x\colon   \wedge^\bullet(\g^*)\to \Obs, \quad \Xi \to \mc O_x^{\Xi}. 
$$
Let $d_{\mr{CE}}: \wedge^\bullet(\g^*)\to \wedge^\bullet(\g^*)$ denote the Chevalley--Eilenberg differential. It is a good exercise (using $\delta_{CS}\A^0={1\over 2}[\A^0, \A^0]$) to show that $\mc O_x$ is in fact a cochain map 
$$
\mc O_x\colon   \bracket{\wedge^\bullet(\g^*), d_{\mr{CE}}}\to  \bracket{\Obs, \delta_{CS}=\{CS,-\}_{\mr{BV}}}. 
$$
In particular, we find that $\mc O_x$ maps Lie algebra cohomologies to $\delta_{CS}$-closed observables. 

As another example, let $S^1$ be a circle inside $M$. Consider a representation $R$ of $\g$ and let $\Tr_R \colon \g \to \C$ be the trace of an element of $\g$ in a representation $R$. We define the following linear observable $\oint_{S^1}^R $ by
$$
\oint_{S^1}^R \colon \E \to \C, \quad \A \to  \int_{S^1} \Tr_R \A^1.  
$$
Let us compute its BV bracket with $CS$. We find 
$$
\fbracket{CS, \oint_{S^1}^R }_{\mr{BV}}(\A)= \oint_{S^1}^R  \delta_{CS} \A^1= \int_{S^1} \Tr_R\bracket{d\A^0+[\A^0, \A^1]}=0. 
$$
In particular, $\oint_{S^1}^R$ is a $\delta_{CS}$-closed observable. 

\end{eg}

\subsection{Homotopy transfer and effective theory}\label{sec:HT} Let $(E, Q, \omega)$ be our infinite-dimensional $(-1)$-shifted symplectic space.  We consider the cohomology
$$
 \HH:= H(\E, Q).
$$
Our assumption $(\E, Q)$ being an elliptic complex implies that $\HH$ is a finite-dimensional graded vector space. Since $\omega$ is compatible with $Q$, $\omega$ descends to define a $(-1)$-shifted symplectic pairing $\omega_{\HH}$ on $\HH$. In other words, we end up with a finite-dimensional $(-1)$-shifted symplectic space $(\HH, \omega_{\HH})$.

 Let
$$
\OO(\HH):= \prod_{k\geq 0}\Hom(\HH^{\otimes k}, \C)_{S_k}
$$
be the space of formal functions on $\HH$. Let
$$
   K_{\HH}=\omega_{\HH}^{-1}\in \Sym^2(\HH)
$$
be the BV kernel on the cohomology. The Poisson kernel $ K_{\HH}$ defines a BV bracket
$$
  \fbracket{-,-}_{\HH}\colon \OO(\HH)\times \OO(\HH)\to \OO(\HH).
$$

 Let  $I\in \Ol(\E)$ be a solution of classical master equation. In this section, we explain how such data can be transferred to define a formal function $I_{\HH}\in \OO(\HH)$  that satisfies the classical master equation $\{I_{\HH}, I_{\HH}\}_{\HH}=0$.

 \begin{defn}\label{defn:HT} Let us consider $\HH:= H(\E, Q)$ equipped with the zero differential.  A \emph{contraction data} $(i,\pi, \hat P)$ of $(\E, Q)$
 	\begin{equation*}
			\hat P\curved (\E, Q)\overset{\pi}{\underset{i}\rightleftarrows} (\HH, 0)
		\end{equation*}
 consists of cochain maps
 $$
    i\colon \HH\to \E, \quad \pi\colon \E\to \HH
 $$
 and a linear map $\hat P$ such that $\pi\circ i=1_{\HH}$ is the identity map on $\HH$ and
 $$
    1_{\E}=i\circ \pi+ \bbracket{Q, \hat P}.
 $$
We require that
 $$
  \hat P\circ \hat P=0, \quad \pi \circ \hat P=0, \quad \hat P\circ i=0.
 $$
 \end{defn}

A contraction data can be constructed by Hodge theory as follows. Suppose we have a linear operator
$$
Q^{\mr{GF}}: \E \to \E
$$ such that $D=[Q, Q^{\mr{GF}}]$
 is a generalized Laplacian. Such $Q^{\mr{GF}}$ is usually called a gauge fixing operator. Let us consider the space
 $$
    \mathbb{H}:= \{\alpha \in  \E \mid Q\alpha=0,\; Q^{\mr{GF}}\alpha=0\}
 $$
of harmonic elements. Hodge theory implies the isomorphism
 $$
   \HH= H(\E, Q)\simeq \mathbb{H}.
 $$

 Let $h_t$ be the integral kernel of the heat operator $e^{-t [Q, Q^{\mr{GF}}]}$. Then we obtain a contraction data via
 \begin{itemize}
 \item
 $
  i\colon \HH \stackrel{\iso}{\to} \mathbb{H} \to \E
 $ is the embedding via harmonic representatives.
 \item $\pi\colon \E\to \mathbb{H}\stackrel{\iso}{\to}\HH $ is the harmonic projection.
 \item $\hat P=\int_0^\infty Q^{GF}e^{-t [Q, Q^{\mr{GF}}]}dt$.
  \end{itemize}

Note that the operator $\hat P$ has an integral kernel which is given by
$$
P= \int_0^\infty (Q^{\mr{GF}}\otimes 1)h_t dt.
$$
This is precisely the propagator of our BV theory.

 We define the following formal function on $\E$
\begin{equation}\label{eqn:tree}
W_{\mr{tree}}(I, P):=\sum_{\Gamma: \mr{trees}}\frac{1}{|\Aut(\Gamma)| } w_\Gamma(I,P)\tag{T}
\end{equation}
Here we sum over all possible connected tree graph $\Gamma$. The functional $w_\Gamma(I,P)$ is the Feynman graph integral with vertex being assigned $I$, internal edge of propagator being $P$, and external edge being the input of a function. $|\Aut(\Gamma)|$ is the size of the automorphism group of $\Gamma$.   For more details, we refer to \cite[Chapters 2 and 5]{CostelloBook}.

 \begin{defn}\label{transfer-formula} Given the contraction data $(i,\pi, P)$ as above, we define the \emph{homotopic transfer} $I_{\HH}$ of a local function $I$ as the formal function on $\HH$ defined by
 $$
    I_{\HH}=\left. W_{\mr{tree}}(I, P) \right|_{\HH} \in \OO(\HH).
 $$
 \end{defn}

 \begin{prop}\label{prop:HT} Given a solution $I$ of the classical master equation and a contraction data, its homotopic transfer $I_{\HH}$ satisfies the following classical master equation on $\HH$
 $$
   \{I_{\HH}, I_{\HH}\}_{\HH}=0.
 $$
 \end{prop}

\begin{proof} This is a standard homotopy transfer theorem.  See for example \cite{CostelloBook,KS-torus}.
\end{proof}

\begin{eg}[Chern--Simons theory]
		We explain the content of homotopy transfer via the example of Chern--Simons theory as discussed in Example \ref{eg:CS}. The space of fields is
		$$
		 \E:= \Omega^\bullet(M)\otimes \g [1]
		$$
and	we will adopt the same conventions there.

		Let us choose a metric on $M$.  Let
		$
		d^*\colon  \mathcal E\to \mathcal E
		$
		be the adjoint of $d$. The Laplacian is
		$$
		D:= dd^*+d^*d.
		$$
		$\mathbb H=\ker \Delta\subset \mathcal E$ is the subspace of harmonics. This leads to a contraction data
		\begin{equation*}
			\hat P\curved (\mathcal E,d)\overset{\pi}{\underset{i}\rightleftarrows} (\mathbb H,0)
		\end{equation*}
		Here $\hat P=-d^*{1\over D}$ where ${1\over D}$ is the Green's operator on forms. $P$ is precisely the propagator of Chern--Simons theory. The homotopy transfer defines an $L_\infty$-structure on $\mathbb H$. The tree formula in Definition \ref{transfer-formula} gives the tree level Feynman diagrams for Chern--Simons theory, and the transferred $L_\infty$-structure is a way to organize the structure of effective theory on zero modes.
\end{eg}

 \subsection{Compactification}\label{sec:comp}

In this section, we generalize the homotopy transfer method and consider a situation where one has a classical BV theory on a product manifold $M=X\times Y$. We discuss its compactification on $X$, that is, what happens after performing homotopic transfer along $X$. This models the physics process of integrating out massive modes on $X$ to get an effective theory on $Y$. When $Y=\mr{pt}$ is a point, this reduces to the discussion in the previous section.

Let us assume that we have free classical BV theory $(E_X, Q_X)$ on $X$ and $(E_Y, Q_Y)$ on $Y$. For the elliptic complexes $(\E_X,Q_X)$ and $(\E_Y,Q_Y)$, let us consider the  tensor product
$$
  \E_{X\times Y}=\E_X\otimes \E_Y:=\Gamma(X\times Y, E_X\boxtimes E_Y)
$$
with differential $Q_{X\times Y}=Q_X\otimes 1+1\otimes Q_Y$. In the following, we will just write $Q_X$ for $Q_X\otimes 1$ and $Q_Y$ for $1\otimes Q_Y$ for simplicity.

 Assume we have a local $(-1)$-shifted symplectic pairing $\omega$ on $\E_{X\times Y}$ such that it is compatible with both $Q_X$ and $Q_Y$. Then the triple
$$
   (  \E_{X\times Y}, Q_{X\times Y}, \omega)
$$
defines a free BV theory on $X\times Y$.

Let us denote as before
$$
\HH_X= H^\bullet(\E_X, Q_X).
$$
Since $\omega$ is compatible with $Q_X$, it descends to define a symplectic pairing on $\HH_X\otimes \E_Y$, denoted by $\omega_{Y}$.  It is easy to see that the triple
$$
\bracket{\HH_X\otimes \E_Y, Q_Y, \omega_Y}
$$
defines a free BV theory on $Y$. Our next goal is to figure out an interaction that arises naturally from an interaction on $ \E_{X\times Y}$.

 \begin{defn}  A \emph{contraction data} $(i,\pi, \hat P)$ of $(\E_{X\times Y}, Q_{X\times Y})$ relative to $\pi_Y\colon X\times Y \to Y$ consists of cochain maps
 $$
    i\colon \HH_X\otimes \E_Y \to \E_{X\times Y}, \quad \pi\colon \E_{X\times Y}\to \HH_X\otimes \E_Y
 $$
 and a linear operator $\hat P$ on $\E_{X\times Y}$ such that $\pi\circ i=1_{\HH_X\otimes \E_Y}$ and
 $$
    1_{\E_{X\times Y}}=i\circ \pi+ \bbracket{Q_X, \hat P}.
 $$
We require that
 $$
  \hat P\circ \hat P=0, \quad \pi \circ \hat P=0, \quad \hat P\circ i=0.
 $$
 \end{defn}
When $Y$ is a point and $\E_Y=  \R$, we recover Definition \ref{defn:HT}.

Assume $Q^{\mr{GF}}_X$ is a gauge fixing operator on $X$ such that $[Q_X, Q_X^{\mr{GF}}]$ is a generalized Laplacian on $X$. Let $h_t^X$ be the heat kernel of the operator $e^{-t[Q_X,Q_X^{\mr{GF}} ]}$. Let $\HH_X$ be the harmonics  of $\E_X$. Then one can find a contraction data $(i, \pi, P)$ relative to $\pi_Y$ where
 \begin{itemize}
\item $    i\colon \HH_X\otimes \E_Y\to \E_{X\times Y}$ is the harmonic embedding (along $X$).
\item$ \pi\colon\E_{X\times Y}\to \HH_X\otimes \E_Y$ is the harmonic projection (along $X$).
\item $\hat P =\hat P_X \otimes 1_{Y}$ where $\hat P_X= \int_0^\infty Q_X^{\mr{GF}} e^{-t[Q_X,Q_X^{\mr{GF}} ]} dt$ acts on $\E_X$ factor and $1_Y$ is the identity map on $\E_Y$ factor. The integral kernel of $\hat P$ on $X\times Y$ is the integral kernel of $\hat P_X$ on $X$ tensoring with the delta-function on $Y$.
\end{itemize}

 \begin{thm}\label{thm-push} Let $I\in \Ol(\E_{X\times Y})$ satisfy the classical master equation for  $   (  \E_{X\times Y}, Q_{X\times Y}, \omega)
$. Let $(i, \pi, \hat P)$ be a contraction data relative to $\pi_Y\colon X\times Y \to Y$ defined via the gauge fixing operator $Q_X^{GF}$ above, and  $P$ be the integral kernel of $\hat P$. Define
 $$
   \pi_{Y*} I := \left. W_{\mr{tree}}(I, P) \right |_{\HH_X\otimes \E_Y}.
 $$
 Here $W_{\mr{tree}}(I, P)$ is the sum of Feynman tree diagrams \eqref{eqn:tree}. $|_{\HH_X\otimes \E_Y}$ means restricting to ${\HH_X\otimes \E_Y}$. Then $\pi_{Y*}I\in \Ol(\HH_X\otimes \E_Y)$ is a local functional on $Y$ and satisfies the classical master equation for the BV theory $\bracket{\HH_X\otimes \E_Y, Q_Y, \omega_Y}
$.

 \end{thm}
 \begin{proof} The proof is completely parallel to Proposition \ref{prop:HT}.  The locality on $Y$ comes from the fact that the propagator $P$ is a $\delta$-function distribution along $Y\times Y$. \end{proof}

Therefore we obtain a classical interacting BV theory  on $Y$.  The functional $\pi_{Y*}I$ is called in physics the effective theory of $I$ on $Y$ obtained by compactification on $X$. A contraction data relative to $\pi_Y\colon X\times Y \to Y$ gives
$$
\text{CME}  (  \E_{X\times Y}, Q_{X\times Y}, \omega) \to \text{CME} \bracket{\HH_X\otimes \E_Y, Q_Y, \omega_Y}
$$
which plays the role of integrating out massive modes along $X$ (at the classical level).

\section{Dispersionless Integrable Hierarchy}\label{sec:3}
In this section we explain the connection between dispersionless integrable hierarchy in 2d topological field theory and the Kodaira--Spencer gravity (BCOV theory). The infinitely many commutating Hamiltonians arise naturally from the infinite abelian symmetries of Kodaira--Spencer gravity (BCOV theory).

 \subsection{Dispersionless Integrable Hierarchy via Maurer--Cartan Equation}\label{subsection:Dispersionless hierarchy}
We collect some basics on 2d topological field theories that will be used in this paper. We refer to \cite{Dubrovin-book} for a comprehensive review.

Let $H$ be a finite-dimensional graded vector space, which is the small phase space of a 2d topological field theory. The big phase space is $H\llb t\rrb$ where $t$ is a formal variable of degree $2$ representing the gravitational descendant. We use $\abracket{-}_0$ to denote the genus zero correlation functions, which is a graded symmetric function on  $H\llb t\rrb$
  $$
 \abracket{-}_0\colon  \Sym^\bullet (H\llb t\rrb)\to \C.
 $$
For example in the A-model, $H$ could be the de Rham cohomology of $X$, and $\abracket{-}_0$ be the genus zero Gromov--Witten invariants; for another example in the B-model, $H$ could be the Jacobian ring of a holomorphic function and $\abracket{-}_0$ be the genus zero Landau--Ginzburg invariants.
Restricting $\abracket{-}_0$ to the small phase space gives $H$ a structure of graded Frobenius manifold. For simplicity, we assume $H$ is purely even in the following discussion.

 Let us choose a basis $\{\OO_\alpha\}_{\alpha}$ of $H$, where we identify $\OO_1$ as the identity element (or so-called puncture operator) and write $P:=\OO_1$. They give a topological basis $\{\OO_\alpha^{(k)}\}_{k\geq 0,\alpha}$ of $H\llb t\rrb$ where $\OO_\alpha^{(k)}:= t^k \OO_\alpha$ represents the $k$-th gravitational descendant of $\OO_\alpha$.  The  matrix
 $$
 g_{\alpha\beta}:=\abracket{\OO_\alpha \OO_\beta P}_0
 $$
 gives a non-degenerate inner product on $H$. Let $g^{\alpha\beta}$ be its inverse matrix. $g$ will be used to raise and lower the indices. For example,
 $
 \OO^{(k)\alpha}:=\sum_{\beta}g^{\alpha\beta} \OO^{(k)}_\beta
 $.

 Let us write $b^\alpha_k$ for the linear coordinates on $H\llb t\rrb$ corresponding to the basis $\{\OO_\alpha^{(k)}\}_{k\geq 0,\alpha}$. We use the following symbol
 $$
 \aabracket{\OO^{(k_1)}_{\alpha_1} \cdots\OO^{(k_n)}_{\alpha_n}}_0(\bb):= \abracket{\OO^{(k_1)}_{\alpha_1} \cdots\OO^{(k_n)}_{\alpha_n} e^{ \sum_{k,\alpha} b^\alpha_k \OO_\alpha^{(k)}}}_0\quad \quad \bb=\{b^\alpha_k\}
 $$
 which is a formal power series in $b^\alpha_k$'s. In particular, $\F_0:=\aabracket{-}_0$ is the genus 0 partition function. Then
 $$
  \aabracket{\OO^{(k_1)}_{\alpha_1} \cdots\OO^{(k_n)}_{\alpha_n}}_0=\pa_{k_1,\alpha_1}\cdots \pa_{k_n,\alpha_n}\F_0, \quad \text{where}\quad \pa_{k,\alpha}:={\pa\over \pa b^\alpha_k}.
 $$
Because $\abracket{-}_0$ requires at least three insertions, $\F_0$ starts with cubic terms.

 The potential function of the Frobenius manifold is given by restricting $\F_0$ to the small phase space, that is, $\F_0|_{b^\bullet_{>0}=0}$, which solves the WDVV equation. The following topological recursion relation also holds \cite{Witten-gravity}
 \begin{equation}\label{TRR}
 \aabracket{\OO^{(i+1)}_\alpha \OO^{(j)}_\beta \OO^{(k)}_\gamma}_0=\sum_{\sigma}\aabracket{\OO^{(i)}_\alpha \OO^{(0)\sigma}}_0 \aabracket{\OO^{(0)}_\sigma\OO^{(j)}_\beta \OO^{(k)}_\gamma}_0.
 \tag{TRR}
 \end{equation}

 From this data, we can associate a classical dispersionless integrable hierarchy as follows \cite{DW,Witten-moduli,Dubrovin-integrable}. Let us promote $b^\alpha_k$ to be a function $b^\alpha_k(z)$ depending on a variable $z$. Consider the space of local functionals in $b^\alpha_k(z)$'s of the following type
 $$
   \oint dz \ \L(b^\alpha_k, \pa_zb^\alpha_k,\pa_z^2b^\alpha_k,\cdots)
 $$
where $\L$ is a  formal power series in $\pa_z^mb^\alpha_k(z)$'s. Alternately, such space can be defined by the formal differential ring generated by $b^\alpha_k(z)$'s modulo total derivatives.

 We introduce a Poisson bracket by
 \begin{equation*}
     \{b_0^\alpha(z), b_0^\beta(w)\}=g^{\alpha\beta} \pa_z \delta(z-w) 
 \end{equation*}
where $\delta(z-w)$ is the $\delta$-function. We ask it to be only nontrivial for $b^\alpha_k$'s when $k=0$. Then it induces a Lie bracket on the above space of local functionals, which we still denote by $\{-,-\}$. Concretely, given two local functionals $S_i=\oint dz\  \mc L_i$ as above, their bracket is 
\begin{equation}\label{PSBracket}
\{S_1, S_2\}:=\sum_{\alpha, \beta, k,m}g^{\alpha \beta}\oint dz \bracket{\bracket{-{\pa\over \pa z}}^k {\pa \mc L_1\over \pa (\pa_z^k b_0^\alpha)}}{\pa\over \pa z} \bracket{\bracket{-{\pa\over \pa z}}^m{\pa \mc L_2 \over \pa (\pa_z^mb_0^\beta)}}. \tag{P}
\end{equation}
One can show that $\{-,-\}$ is well-defined on the differential ring modulo total derivatives. For example, using 
\begin{equation*}
\pa_z\frac{\pa}{\pa(\pa_z^kb_0^{\alpha})}=\frac{\pa}{\pa(\pa_z^kb_0^{\alpha})}\pa_z-\frac{\pa}{\pa(\pa_z^{k-1}b_0^{\alpha})}
\end{equation*}
we can check it is well-defined as follows (this is basically a D-module manipulation)
\begin{align*}
&\fbracket{\oint dz\ \pa_z\L_1, \oint dz\ \L_2}\\
=~&\sum_{\alpha, \beta, k,m}g^{\alpha \beta}\oint dz\ \bracket{\bracket{-{\pa\over \pa z}}^k {\pa (\pa_z\L_1)\over \pa (\pa_z^k b_0^\alpha)}}{\pa\over \pa z} \bracket{\bracket{-{\pa\over \pa z}}^m{\pa \mc L_2 \over \pa (\pa_z^mb_0^\beta)}}\\
=~&-\sum_{\alpha, \beta, k,m}g^{\alpha \beta}\oint dz\ \bracket{\bracket{-{\pa\over \pa z}}^{k+1} {\pa \L_1\over \pa (\pa_z^k b_0^\alpha)}}{\pa\over \pa z} \bracket{\bracket{-{\pa\over \pa z}}^m{\pa \mc L_2 \over \pa (\pa_z^mb_0^\beta)}}\\
&+\sum_{\alpha, \beta, k,m}g^{\alpha \beta}\oint dz\ \bracket{\bracket{-{\pa\over \pa z}}^k {\pa \L_1\over \pa (\pa_z^{k-1} b_0^\alpha)}}{\pa\over \pa z} \bracket{\bracket{-{\pa\over \pa z}}^m{\pa \mc L_2 \over \pa (\pa_z^mb_0^\beta)}}=0.
\end{align*}
It can be further checked directly that $\{-,-\}$ gives a Lie bracket.


Let $G_{k,\alpha}$ be the restriction of 1-point functions to the small phase space
$$
G_{k,\alpha}(b_0^\bullet)= \left. \aabracket{\OO^{(k)}_\alpha}_0(\bb)\right|_{b^\bullet_{>0}=0}.
$$
The key observation is the following proposition (see \cite{Witten-moduli,Dubrovin-integrable}).

\begin{prop} \label{prop:dispersionless hierarchy}
The local functionals
 $\{
\oint dz\ G_{k,\alpha}(b_0^\bullet)\}
$
commute with each other:
 $$
 \fbracket{\oint dz\ G_{k,\alpha}(b_0^\bullet), \oint dz\ G_{m,\beta}(b_0^\bullet)}=0, \quad \forall \alpha, \beta, \quad \text{and}\quad \forall k, m\geq 0.
 $$
 \end{prop}


Such infinitely many commuting Hamiltonians  $\{ \oint dz\ G_{k,\alpha}(b_0^\bullet)\}
 $ generate the classical dispersionless integrable hierarchy. Below we will interpret and derive this proposition from a completely different perspective. In the next sections, we will explain the origin of our construction from the viewpoint of Kodaira--Spencer gravity (BCOV theory).

\begin{eg}[Dispersionless KdV] Consider $H=\C$ and $\abracket{-}_0$ comes from pure gravity. It is known that
\begin{equation*}
\abracket{P^{(k_1)}\cdots P^{(k_n)}}_0=\int_{\overline{M}_{0,n}}\psi_1^{k_1}\cdots\psi_n^{k_n}=\binom{n-3}{~k_1,\cdots,k_n~}.
\end{equation*}
Then we have
$$G_k(b_0)=\abracket{P^{(k)}e^{b_0P}}_0=\frac{1}{(k+2)!}b_0^{k+2}\abracket{P^{(k)}\overbrace{P\ \cdots\ P}^{k+2}}_0=\frac{1}{(k+2)!}b_0^{k+2}.$$
$\left\{\oint dz\ G_k(b_0)\right\}_{k\geq0}$ give the commuting Hamiltonians of the dispersionless KdV hierarchy.
Indeed, by the definition of the bracket,
$$\fbracket{\oint dz\ G_k(b_0), \oint dz\ G_{m}(b_0)}=\frac{1}{(k+1)!m!}\ \oint dz\ b_0^{k+1+m}\pa_zb_0=0.$$
\end{eg}

We enlarge the above bosonic fields $b^\alpha_k(z)$'s by introducing additional fermions $\eta^\alpha_k(z)$'s. Let $\Xi $ denote the space of local functionals of $b^\alpha_k(z), \eta^{\alpha}_k(z)$'s, i.e. the differential ring generated by $b^\alpha_k(z), \eta^{\alpha}_k(z)$ modulo total derivatives. An element of $\Xi $ can be expressed by
$$
\oint dz\ \mathcal L(b^\alpha_k, \eta^{\alpha}_k,\pa_z b^\alpha_k, \pa_z \eta^{\alpha}_k,\cdots)
$$
where $\mathcal L$ is a formal power series in derivatives of $b^\alpha_k(z), \eta^{\alpha}_k(z)$'s. Note that $\eta^\alpha_k$'s anti-commute with each other, and $\Xi$ is a $\Z/2\Z$-graded space. The above defined Poisson bracket on $b^\alpha_0$'s extends to define a bracket on $\Xi $, by the same formula \eqref{PSBracket}. 

We introduce a differential $\delta$ by
$$
\delta b^\alpha_k=\pa_z  \eta^{\alpha}_{k-1} , \quad \delta \eta^{\alpha}_k=0.
$$
Since $k\geq 0$, the above formula reads $\delta b^\alpha_0=0$ when $k=0$.

\begin{prop}\label{prop-dgla}The triple $(\Xi , \delta, \{-,-\})$ forms a differential $\Z/2\Z$-graded Lie algebra.
\end{prop}
\begin{rmk} It is in fact $\Z$-graded, but we won't need this fact in this paper. 
\end{rmk}
\begin{proof} 
We only need to prove that the differential $\delta$ is compatible with the bracket $\{-,-\}$. It suffices to consider the densities of the form
$$\L_i(b^\alpha_k,\eta^{\alpha}_k, \pa_z b^\alpha_k, \pa_z\eta^{\alpha}_k,\cdots)=f_i(b^\alpha_0, \pa_z b^\alpha_0,\cdots)g_i(b^\alpha_{k>0},\eta^{\alpha}_k, \pa_z b^\alpha_{k>0},\pa_z\eta^{\alpha}_k,\cdots)\quad i=1,2.$$
The compatibility follows from the fact that $\delta$ commutes with $\pa_z$: 
\begin{align*}
&\delta\left\{\oint\ dz\ \L_1(b^\alpha_k,\eta^{\alpha}_k, \pa_z b^\alpha_k, \pa_z\eta^{\alpha}_k,\cdots),\oint dz\ \L_2(b^\alpha_k,\eta^{\alpha}_k, \pa_z b^\alpha_k, \pa_z\eta^{\alpha}_k,\cdots)\right\}\\
=~&\delta\oint dz\ \sum_{\alpha, \beta, k,m}\left(-\frac{\pa}{\pa z}\right)^k\left(\frac{\pa f_1}{\pa(\pa^kb_0^{\alpha})}g_1\right)g^{\alpha\beta}\pa_z\left(-\frac{\pa}{\pa z}\right)^m\left(\frac{\pa f_2}{\pa(\pa^mb_0^{\alpha})}g_2\right)\\
=~&\oint dz\ \sum_{\alpha, \beta, k,m}\left(-\frac{\pa}{\pa z}\right)^k\left(\frac{\pa f_1}{\pa(\pa^kb_0^{\alpha})}\delta g_1\right)g^{\alpha\beta}\pa_z\left(-\frac{\pa}{\pa z}\right)^m\left(\frac{\pa f_2}{\pa(\pa^mb_0^{\alpha})}g_2\right)\\
&+(-1)^{|\L_1|}\oint dz\ \sum_{\alpha, \beta, k,m}\left(-\frac{\pa}{\pa z}\right)^k\left(\frac{\pa f_1}{\pa(\pa^kb_0^{\alpha})}g_1\right)g^{\alpha\beta}\pa_z\left(-\frac{\pa}{\pa z}\right)^m\left(\frac{\pa f_2}{\pa(\pa^mb_0^{\beta})}\delta g_2\right)\\
=~&\left\{\oint\ dz\ \delta \L_1, \oint dz\ \L_2\right\}+(-1)^{|\L_1|}\left\{\oint\ dz\ \L_1, \oint dz\ \delta\L_2\right\}.\qedhere
\end{align*} 
\end{proof}

Define the following functional $I\in \Xi$ by
\[I=\sum_{k,\alpha}\oint dz \  \eta^\alpha_k(z) \pa_{k,\alpha}\F_0(\bb(z)), \quad \bb=\{b^\alpha_k\}.\]
Equivalently, we can write
\[I=\sum_{k,\alpha}\oint dz \ \eta^\alpha_k \aabracket{\OO^{(k)}_\alpha}_0(\bb).\]

\begin{thm}\label{thm-MC}
The functional $I$ satisfies the following Maurer--Cartan equation (MC)
$$
\delta I+{1\over 2}\{I, I\}=0.
$$
In fact, this equation is equivalent to the topological recursion relation for $\F_0$.

\end{thm}
\begin{proof}We have
\begin{align*}
\delta I=\sum_{k,\alpha,m,\beta}&\oint dz\ \pa_z\eta^\beta_m  \eta^\alpha_k \aabracket{\OO^{(m+1)}_\beta\OO^{(k)}_\alpha}_0(\bb)
\end{align*}
and
\begin{align*}
{1\over 2}\{I, I\}&=\sum_{k,\alpha,m,\beta, \delta}{1\over 2}\oint dz\ \bracket{ \eta^\alpha_k \aabracket{\OO^{(k)}_\alpha \OO^{(0)\delta}}_0(\bb)}{\pa\over \pa z} \bracket{ \eta^\beta_m \aabracket{\OO^{(0)}_\delta\OO^{(m)}_\beta}_0(\bb)}.
\end{align*}
Note that both $\delta I$ and ${1\over 2}\{I, I\}$ are quadratic in $\eta$'s. Hence to show that their sum are equal to zero, it is enough to check its variation with respect to $\eta$ vanishes. Taking into account the fermionic property of $\eta$ and integration by part, we find
\begin{align*}
{\delta (\delta I)\over \delta \eta^\alpha_k(z)}=&-\sum_{\beta}\pa_z\eta^\beta_m\aabracket{\OO^{(m+1)}_\beta\OO^{(k)}_\alpha}_0(\bb)-\sum_{\beta}{\pa\over \pa z}\bracket{\eta^\beta_m\aabracket{\OO^{(k+1)}_\alpha\OO^{(m)}_\beta}_0(\bb)}\\
=&-\sum_{\beta}\pa_z \eta_m^\beta\aabracket{\OO^{(m+1)}_\beta\OO^{(k)}_\alpha+\OO^{(k+1)}_\alpha\OO^{(m)}_\beta}_0(\bb)-\sum_{\beta,\gamma}\eta_m^\beta \pa_z b^\gamma_l \aabracket{\OO^{(k+1)}_\alpha\OO^{(m)}_\beta \OO^{(l)}_\gamma}_0(\bb).
\end{align*}
Similarly,
\begin{align*}
{\delta \bracket{{1\over 2}\{I, I\}}\over \delta  \eta^\alpha_k(z)}=&\sum_{\delta, \beta}\aabracket{\OO^{(k)}_\alpha \OO^{(0)\delta}}_0(\bb){\pa\over \pa z} \bracket{ \eta^\beta_m \aabracket{\OO^{(0)}_\delta\OO^{(m)}_\beta}_0(\bb)}\\
=&~ \sum_{\delta, \beta}\pa_z \eta_m^\beta \aabracket{\OO^{(k)}_\alpha \OO^{(0)\delta}}_0(\bb)\aabracket{\OO^{(0)}_\delta\OO^{(m)}_\beta}_0(\bb)\\
& +\sum_{\delta, \beta}\eta_m^\beta \pa_z b^\gamma_l \aabracket{\OO^{(k)}_\alpha \OO^{(0)\delta}}_0(\bb) \aabracket{\OO^{(0)}_\delta\OO^{(m)}_\beta \OO^{(l)}_\gamma}_0(\bb).
\end{align*}

Comparing the coefficients of $\eta^\beta_m$ and $\pa_z \eta^\beta_m$, we see that $\delta I+{1\over 2}\{I,I\}=0$ is equivalent to
\begin{align*}
\aabracket{\OO^{(k+1)}_\alpha\OO^{(m)}_\beta \OO^{(l)}_\gamma}_0=\sum_{\delta}\aabracket{\OO^{(k)}_\alpha \OO^{(0)\delta}}_0 \aabracket{\OO^{(0)}_\delta\OO^{(m)}_\beta \OO^{(l)}_\gamma}_0 \\
\aabracket{\OO^{(m+1)}_\beta\OO^{(k)}_\alpha+\OO^{(k+1)}_\alpha\OO^{(m)}_\beta}_0=\sum_{\delta}\aabracket{\OO^{(k)}_\alpha \OO^{(0)\delta}}_0\aabracket{\OO^{(0)}_\delta\OO^{(m)}_\beta}_0.
\end{align*}
The first equation is precisely TRR. To show the second equation, we take its derivative, which is TRR, and note that both sides cannot have a constant to differ.
\end{proof}

The above connection between TRR and MC equation gives a very simple way to understand Proposition \ref{prop:dispersionless hierarchy}.

\begin{proof}[Proof of Proposition \ref{prop:dispersionless hierarchy}]
Let us restrict our fields to the following locus
$$
b^\alpha_{k>0}=0, \quad \eta^{\alpha}_k=\text{constant}.
$$
This is called the stationary sector. Let $I_S$ denote the restriction of $I$ to this sector. Since $\delta=0$ in the stationary sector, Maurer--Cartan equation implies
$$
\{I_S, I_S\}=0.
$$
Observe that
$$
  I_S=\sum_{k,\alpha}\oint dz\ \eta^{\alpha}_k G_{k,\alpha}(b_0^\bullet).
  $$
Since $\eta^{\alpha}_k$'s are fermionic variables,
$$
\{I_S, I_S\}=\sum_{k,\alpha;m,\beta} \eta^\alpha_k \eta^\beta_m \fbracket{\oint dz\ G_{k,\alpha}, \oint dz\ G_{m,\beta} }=0
$$
which is equivalent to the commutativity of the local functionals $\{\oint dz\ G_{k,\alpha}\}$.
\end{proof}

 \subsection{A generalization of BCOV theory}
 Let $Y$ be a  Calabi--Yau manifold with holomorphic volume form $\Omega_Y$. B-model on $Y$ can be described by the Kodaira--Spencer gravity \cite{BCOV}. We will work with the formulation in \cite{Si-Kevin} that generalizes the Kodaira--Spencer gravity by including the gravitational descendant. Such generalization is called BCOV theory in \cite{Si-Kevin}. It is precisely in this formulation that integrable hierarchy will appear naturally.

 The fields in BCOV theory is given by the cochain complex
 $$
     \PV(Y)\llb t\rrb, \quad Q=\dbar+t\pa.
 $$
 Here
 $$
 \PV(Y)=\bigoplus_{i,j}\PV^{i,j}(Y), \quad \PV^{i,j}(Y):=\Omega^{0,j}(Y, \wedge^i T_Y^{1,0})
 $$
 is the space of smooth polyvector fields. $\PV^{i,j}(Y)$ sits in degree $i+j$ and $t$ is a formal variable of degree 2 representing the gravitational descendant. Here $\pa$ is the divergence operator on $\PV(Y)$ with respect to $\Omega_Y$: that is, if we identify $\PV(Y)$  with differential forms on $Y$ via contracting $\Omega_Y$
 $$
   \PV^{\bullet, \bullet}(Y)\stackrel{\vdash \Omega_Y}{\to} \Omega^{d-\bullet, \bullet}(Y),\quad d=\dim_{\C}(Y),
 $$
then $\pa$ corresponds to the holomorphic de Rham differential, hence justifying the notation. In \cite{Si-Kevin}, a degenerate BV structure together with a solution of classical master equation is found. This builds up the classical BCOV theory.

In this section, we generalize this set-up and introduce a notion of $H$-valued BCOV theory. Here $H$ is the state space as in the previous subsection and we will adopt the same notations used there. Let
 $$
  \PV(Y;H):= \PV(Y)\otimes H, \quad \PV(Y;H)\llb t\rrb= (\PV(Y)\otimes H)\llb t\rrb.
 $$

 Let us denote the following trace map by integrating out polyvectors
 $$
 \Tr: \PV_c(Y)\to \C, \quad \mu\to \int_Y (\mu\vdash \Omega_Y)\wedge \Omega_Y.
 $$
 The subscript $c$ refers to ``compactly supported". 

 The BV kernel of $H$-valued BCOV theory is degenerate, given by
 $$
   K_0= (g^{\alpha\beta}\OO_\alpha\otimes \OO_\beta)(\pa\otimes 1)\delta_Y
 $$
 where $\delta_Y$ is the ($\PV(Y\times Y)$-valued) delta distribution supported on the diagonal of $Y\times Y$. With respect to the pairing
 $$
\PV_c(Y)\otimes \PV_c(Y)\to \C,\quad \mu_1\otimes \mu_2\to \Tr(\mu_1\wedge \mu_2),
 $$
 $(\pa\otimes 1)\delta_Y  $ is the integral kernel of the operator $\pa\colon \PV(Y)\to \PV(Y)$.

$K_0$ is a $(-1)$-shifted Poisson kernel, defining the degree $1$ BV bracket
 $$
   \{-,-\}_{\mr{BV}}\colon \Ol(\PV(Y;H)\llb t\rrb)\times \Ol(\PV(Y;H)\llb t\rrb)\to \Ol(\PV(Y;H)\llb t\rrb).
 $$
 Here $\Ol(\PV(Y;H)\llb t\rrb)$ is the space of local functionals on $\PV(Y;H)\llb t\rrb$.  The triple $(\PV(Y;H)\llb t\rrb, Q, K^H_0)$ defines a degenerate BV data, and the formalism in previous section works for this degenerate situation in a similar fashion. For a general framework of degenerate BV theory at the classical level, one may refer to \cite{ButsonYoo}.

\begin{defn}\label{vertex contribution}We define the \emph{classical BCOV interaction} $I^H_Y$ by
$$
  I^H_Y(\mu)=\Tr \abracket{e^\mu}_0,\quad \mu \in \PV_c(Y;H)\llb t\rrb
$$
where
$$
\abracket{\mu_1\otimes t^{k_1}\OO_{i_1}, \cdots, \mu_n\otimes t^{k_n}\OO_{i_n}}_0:=\abracket{\OO_{i_1}^{(k_1)}\cdots  \OO_{i_n}^{(k_n)}}_0 \mu_1\wedge\cdots \wedge \mu_n, \quad \mu_i\in \PV_c(Y).
$$
\end{defn}
\begin{prop}\label{prop-CME}
The interaction term $I_Y^H$ satisfies the classical master equation
$$
  QI_Y^H+{1\over 2}\fbracket{I_Y^H, I_Y^H}_{\mr{BV}}=0.
$$
\end{prop}
\begin{proof} This proof is completely the same as that in \cite[Lemma 2.10.2]{Si-Kevin} \cite[Lemma 4.6]{thesis} which only uses the genus 0 topological recursion relation.
\end{proof}

When $H=\C$ and $\abracket{-}_0$ comes from the pure gravity, this reduces to the classical BCOV theory as described in \cite{Si-Kevin}.

\subsection{Current observables and integrable hierarchies}\label{sec:current-integrable}
We consider the situation when $Y$ is one-dimensional Calabi--Yau, i.e.,
$$
  Y=\C, \quad \C^\times,\quad  \text{or} \quad E_\tau:=\C/(\Z\oplus \Z\tau).
$$
We will work with $Y=\C^\times$ in this paper, which will be convenient for the purpose of integrable hierarchy. Let us first fix some notations. Let $z$ be the linear holomorphic coordinate on $\C$. We write $\C^\times$ as the quotient of $\C$ by identifying
$$
z\sim z+1.
$$
The Calabi--Yau structure on $\C^\times$ is given by the holomorphic 1-form $dz$.

The field content of the $H$-valued BCOV  theory on $\C^\times$ is
$$
\E^H:=(\PV(\C^\times)\otimes H)\llb t\rrb=\bracket{\Omega^{0,\bullet}(\C^\times)\otimes H}\llb t\rrb \oplus \bracket{\Omega^{0,\bullet}(\C^\times, T_{\C^\times}[-1])\otimes H} \llb t\rrb.
$$
We will represent an arbitrary field $\mu$ in components by
$$
\mu=\sum_{k\geq 0, \alpha}
\bracket{\lambda^\alpha_k\otimes \OO_\alpha t^k + \rho^\alpha_k \pa_z \otimes \OO_\alpha t^k}
$$
where
$$
  \lambda^\alpha_k\in \Omega^{0,\bullet}(\C^\times), \quad \rho^\alpha_k \in  \Omega^{0,\bullet}(\C^\times).
$$
With respect to the trace pairing, the $(\PV(\C^\times)\otimes \PV(\C^\times))$-valued $\delta$-distribution $\delta_{\C^\times}$ is 
$$
   \delta_{\C^\times}= \delta(z-w) (d\bar z-d\bar w) (\pa_z-\pa_w).
$$
Here $\delta(z-w)$ is the $\delta$-function on $\C^\times$ normalized by
$$
\int_{\C^\times} f(z,\bar z)\delta(z-w)dz\wedge d\bar z= f(w,\bar w), \qquad \forall  f\in C^\infty(\C^\times).
$$
The $\delta$-distribution $\delta_{\C^\times}$ has the defining property that
$$
\Tr_{w\in \C^\times} (\delta_{\C^\times}(z-w)\mu(w,\bar w,\pa_w, d\bar w))=\mu(z,\bar z, \pa_z, d\bar z), \qquad \forall \mu \in \PV(\C^\times).
$$
The BV kernel $K_0$ is therefore
$$
K_0=\sum_{\alpha\beta}(g^{\alpha\beta}\OO_\alpha\otimes \OO_\beta) (\pa\otimes 1)\delta_{\C^\times}=\sum_{\alpha}(\OO^\alpha\otimes \OO_\alpha) \pa_z \delta(z-w) (d\bar z-d\bar w) .
$$

Let us consider the following linear maps
$$
\E^H \to \Omega^{0,\bullet}(\C^\times), \quad \mu \mapsto \pa_z^m \lambda^\alpha_k.
$$
By abuse of notation, we will just call this map $\pa_z^m \lambda^\alpha_k$. Similarly we have $\pa_z^m \rho^\alpha_k$.   These linear maps will be denoted by
$$
\pa_z^m \lambda^\alpha_k, \pa_z^m \rho^\alpha_k: \quad \E^H \to \Omega^{0,\bullet}(\C^\times).
$$
Note that $\pa_z^m \lambda^\alpha_k$'s are even maps and $\pa_z^m \rho^\alpha_k$'s are odd maps.

By abuse of notation, we let $\C\llb \pa_z^m \lambda^\alpha_k,\pa_z^m \rho^\alpha_k\rrb $ denote the power series  in $\{\pa_z^m \lambda^\alpha_k, \pa_z^m\rho^\alpha_k\}$'s, where $\pa_z^m \lambda^\alpha_k$'s are  even elements and $\pa_z^m \rho^\alpha_k$'s are odd elements. This is a formal differential ring generated by the symbols $\{\lambda^\alpha_k, \rho^\alpha_k\}$'s. Given $\mc L\in \C\llb \pa_z^m \lambda^\alpha_k,\pa_z^m \rho^\alpha_k\rrb $, we define the $\Omega^{\bullet,\bullet}(\C^\times)$-valued formal functions $\J_{\L}$ on the field space $\E^{H}$ by
$$
\J_{\L}:=dz\ \mc L(\pa_z^m \lambda^\alpha_k, \pa_z^m \rho^\alpha_k).
$$

Introduce a differential $\delta$ on $\C\llb \pa_z^m \lambda^\alpha_k,\pa_z^m \rho^\alpha_k\rrb $ by
$$
\delta \lambda^\alpha_k=\pa_z  \rho^{\alpha}_{k-1} , \quad \delta \rho^{\alpha}_k=0.
$$
Note that this is of the same form as the one defined in Subsection \ref{subsection:Dispersionless hierarchy}. The following formula holds (by unpacking various definitions)
$$
Q \J_{\L}= \dbar(\J_{\L})+\J_{\delta \L}.
$$
Here the differential $Q$ on $\E^H$ induces a differential by duality on various functions on $\E^H$ (here is $\Omega^{\bullet,\bullet}(\C^\times)$-valued function), in the same fashion as defined in Subsection \ref{sec:free-BV}. This gives $Q \J_{\L}$. The second term $\dbar(\J_{\L})$ is the composition of $\J_{\L}$ with the operator $\dbar\colon  \Omega^{0,\bullet}(\C^\times)\to \Omega^{0,\bullet}(\C^\times)$. Intuitively, this may be thought of as having a cochain map $(\C\llb \pa_z^m \lambda^\alpha_k,\pa_z^m \rho^\alpha_k\rrb,\delta )\to \left(\mr{Map} \left( ( \mc E^H ,\dbar), (\Omega^{0,\bullet}(\C^\times), \dbar ) \right), t\pa \right)$.

We will write
$$
\J_{\L}:= \J_{\L}^{(1)}+ \J_{\L}^{(2)}
$$
where $\J_{\L}^{(k)}$ picks up the $k$-form part. Writing the above formula in components says
$$
Q\J_{\L}^{(1)}= \J_{\delta \L}^{(1)}, \quad Q \J_{\L}^{(2)}=\dbar \J_{\L}^{(1)}+ \J_{\delta \L}^{(2)}.
$$
Note that $\dbar \J_{\L}^{(1)}=d \J_{\L}^{(1)}$ by the type reason, so $\J_{\L}^{(2)}$ can be viewed as the topological descendant \cite{Witten-TQFT} of $\J_{\L}^{(1)}$.   Let us fix the circle
$$
 C=\{z\in [0,1]\}\subset \C^\times.
$$
This leads to two observables
$$
   \oint_C \J_{\L}^{(1)}, \quad \int_{\C^\times} \J_{\L}^{(2)}.
$$
We often call $\oint_C \J_{\L}^{(1)}$ a \emph{current observable}.

Recall from Subsection \ref{sec:free-BV} that any local functional $S$ defines a derivation $\{S,-\}_{\mr{BV}}=\delta_S$ on observables.  We will be only interested in local functionals of the form $\int_{\C^\times} \J_{\L}^{(2)}$ defined above. In this case we could have a rather explicit description of such BV bracket.

\begin{prop}\label{prop-bracket} Let $\L_1, \L_2\in \C\llb \pa_z^m \lambda^\alpha_k, \pa_z^m\rho^\alpha_k\rrb $. Denote
$$
\bbracket{\L_1, \L_2}:=\sum_{\alpha,\beta,k,m}\fbracket{\bracket{-{\pa\over \pa z}}^k {\pa \L_1\over \pa (\pa_z^k \lambda_0^\alpha)}}~g^{\alpha\beta}~{\pa\over \pa z} \fbracket{ \bracket{-{\pa\over \pa z}}^m {\pa \L_2\over \pa (\pa_z^m \lambda_0^\beta)} } \in  \C\llb \pa_z^m \lambda^\alpha_k, \pa_z^m\rho^\alpha_k\rrb.
$$
Then $\fbracket{\int_{\C^\times} \J_{\L_1}^{(2)},-}_{\mr{BV}}$ acts on observables $  \oint_C \J_{\L_2}^{(1)}$ and  $\int_{\C^\times} \J_{\L_2}^{(2)}$ by
$$
   \fbracket{\int_{\C^\times} \J_{\L_1}^{(2)},\oint_C \J_{\L_2}^{(1)}}_{\mr{BV}}=\oint_C \J_{[\L_1, \L_2]}^{(1)}, \quad \fbracket{\int_{\C^\times} \J_{\L_1}^{(2)},\int_{\C^\times} \J_{\L_2}^{(2)}}_{\mr{BV}}=\int_{\C^\times} \J_{[\L_1,\L_2]}^{(2)}.
$$
\end{prop}

\begin{proof} This is computed directly using the description of BV bracket  in Subsection \ref{sec:free-BV}.
\end{proof}

\begin{cor} Let $\Obs^{\oint_C}$ denote the space of observables of the form $\{\oint_C \J_{\L}^{(1)}\}_{\L \in \C\llb \pa_z^m \lambda^\alpha_k, \pa_z^m\rho^\alpha_k\rrb }$. We define a bracket $\{-,-\}_C$ on $\Obs^{\oint_C}$ by
$$
 \fbracket{\oint_C \J_{\L_1}^{(1)},\oint_C \J_{\L_2}^{(1)}}_C:=\fbracket{\int_{\C^\times} \J_{\L_1}^{(2)},\oint_C \J_{\L_2}^{(1)}}_{\mr{BV}}.
$$
Then $\{-,-\}_C$ defines a Lie bracket on $\Obs^{\oint_C}$. Moreover, consider the differential
$$
 \delta: \Obs^{\oint_C}\to \Obs^{\oint_C}\colon \quad    \oint_C \J_{\L}^{(1)}\to  \oint_C \J_{\delta \L}^{(1)}.
$$
For simplicity, we still denote this differential by $\delta$. Then the triple $\fbracket{\Obs^{\oint_C},\delta, \fbracket{-,-}_C}$ forms a differential $\Z/2\Z$-graded Lie algebra.
\end{cor}
\begin{proof} By Proposition \ref{prop-bracket}, we have
$$
   \fbracket{\oint_C \J_{\L_1}^{(1)},\oint_C \J_{\L_2}^{(1)}}_C =\oint_C \J_{[\L_1, \L_2]}^{(1)}
$$
which still lies in $\Obs^{\oint_C}$. It is easy to see that $\{-,-\}_C$ is skew-symmetric. We check the Jacobi identity holds. By Proposition \ref{prop-bracket}
\begin{align*}
  &\fbracket{\oint_C \J_{\L_1}^{(1)},\fbracket{\oint_C \J_{\L_2}^{(1)},\oint_C \J_{\L_3}^{(1)}}_C}_C\mp \fbracket{\oint_C \J_{\L_2}^{(1)},\fbracket{\oint_C \J_{\L_1}^{(1)},\oint_C \J_{\L_3}^{(1)}}_C}_C\\
  =&\fbracket{\int_{\C^\times} \J_{\L_1}^{(2)},\fbracket{\int_{\C^\times} \J_{\L_2}^{(2)},\oint_C \J_{\L_3}^{(1)}}_{\mr{BV}}}_{\mr{BV}}\mp \fbracket{\int_{\C^\times} \J_{\L_2}^{(2)},\fbracket{\int_{\C^\times} \J_{\L_1}^{(2)},\oint_C \J_{\L_3}^{(1)}}_{\mr{BV}}}_{\mr{BV}}\\
  =& \fbracket{\fbracket{ \int_{\C^\times} \J_{\L_1}^{(2)},  \int_{\C^\times} \J_{\L_2}^{(2)} }_{\mr{BV}},   \oint_C \J_{\L_3}^{(1)}   }_{\mr{BV}}\\
  =&  \fbracket{\int_{\C^\times} \J^{(2)}_{[\L_1, \L_2]} ,   \oint_C \J_{\L_3}^{(1)}   }_{\mr{BV}} =\fbracket{ \fbracket{\oint_C \J_{\L_1}^{(1)},\oint_C \J_{\L_2}^{(1)}}_C ,  \oint_C \J_{\L_3}^{(1)} }_C
\end{align*}
as required. Here we have used the Jacobi identity for the BV bracket. The compatibility of $\delta$ with $\fbracket{-,-}_C$ is similar to Proposition \ref{prop-dgla}. 
\end{proof}

With the above preparation at hand, we can make contact with the discussions in Subsection \ref{subsection:Dispersionless hierarchy}.  Consider the following power series in $\{\lambda^\alpha_k,\rho^\alpha_k\}$'s
$$
\mc L^H=\sum_{k,\alpha} \rho^\alpha_k \abracket{\OO^{(k)}_\alpha \exp\bracket{\sum \lambda^\beta_m \OO_\beta^{(m)}} }_0.
$$
The BCOV interaction is given by
$$
  I_{\C^\times}^H= \int_{\C^\times} \J_{\mc L^H}^{(2)},
$$
which solves the classical master equation
$$Q I_{\C^\times}^H+\frac{1}{2}\fbracket{I_{\C^{\times}}^H, I_{\C^{\times}}^H}_{\mr{BV}}=0.$$

\begin{prop} $\oint_C \J_{\mc L^H}^{(1)}$ satisfies the Maurer--Cartan equation in $\fbracket{\Obs^{\oint_C},\delta, \fbracket{-,-}_C}$
$$
\delta  \oint_C \J_{ \mc L^H}^{(1)}+{1\over 2}\fbracket{\oint_C \J_{\mc L^H}^{(1)},\oint_C \J_{\mc L^H}^{(1)}}_C=0.
$$
\end{prop}
\begin{proof} We expand the classical master equation for $I_{\C^\times}^H$:
\begin{align*}
0=&~Q I_{\C^\times}^H+\frac{1}{2}\fbracket{I_{\C^{\times}}^H, I_{\C^{\times}}^H}_{\mr{BV}}\\
 =&\int_{\C^\times} \J_{\delta\mc  L^H}^{(2)}+\frac{1}{2}\int_{\C^\times} \J_{[\mc L^H,\mc L^H]}^{(2)}=\int_{\C^\times} \J_{\delta\mc L^H+\frac{1}{2}[\mc L^H,\mc L^H]}^{(2)}.
\end{align*}
where we used the observation that the differential $Q$ on local functionals precisely corresponds to the differential $\delta$ on the power series. This shows that $ \J_{\delta\mc L^H+\frac{1}{2}[\mc L^H,\mc L^H]}^{(2)}$ is a total derivative in $\pa_z, \pa_{\bar z}$. Observe that $\delta\mc L^H+\frac{1}{2}[\mc L^H,\mc L^H]$ only contains holomorphic derivatives. This implies that $ \J_{\delta\mc L^H+\frac{1}{2}[\mc L^H,\mc L^H]}^{(1)}$ is a total derivative in $\pa_z$'s only. Therefore, by Proposition \ref{prop-bracket}, we obtain
$$0=\oint_C\J_{\delta\mc L^H+\frac{1}{2}[\mc L^H,\mc L^H]}^{(1)}=\delta  \oint_C \J_{ \mc L^H}^{(1)}+{1\over 2}\fbracket{\oint_C \J_{\mc L^H}^{(1)},\oint_C \J_{\mc L^H}^{(1)}}_C.\qedhere$$
\end{proof}


This proposition is basically equivalent to Theorem \ref{thm-MC}. Here we have made its connection with the classical BV master equation via the induced Poisson bracket on current observables. If we specialize to the part where $\lambda_m^\beta=0$ for $m\geq 1$ and $\rho_k^\alpha$ is constant, called the \emph{stationary sector} and denoted by $\E^H_S$, then we are led to the following proposition.

\begin{prop} Let us define
$$
   \mc G_{k,\alpha}:=  \abracket{\OO^{(k)}_\alpha \exp\bracket{\sum \lambda^\beta_0 \OO_\beta} }_0 \in \C\llb \lambda^\alpha_0\rrb .
$$
Then they give rise to commuting observables with respect to the Lie bracket $\{-,-\}_C$
$$
\fbracket{\oint_C  \J_{\mc G_{k,\alpha}}^{(1)}, \oint_C  \J_{\mc G_{m,\beta}}^{(1)}}_C=0, \quad \forall k,m, \alpha, \beta.
$$
\end{prop}
\begin{proof} This follows from the discussion in Subsection \ref{subsection:Dispersionless hierarchy}. In fact, we can identify $\lambda^\alpha_k, \rho^\alpha_k$ with the fields $b^\alpha_k, \eta^\alpha_k$ in Subsection \ref{subsection:Dispersionless hierarchy} .
$$
\lambda^\alpha_k \Leftrightarrow b^\alpha_k, \quad \rho^\alpha_k\Leftrightarrow \eta^\alpha_k.
$$
Under this identification, the current observable $\oint_C \J_{\mc L^H}^{(1)}$ becomes precisely the function $I$ in Subsection \ref{subsection:Dispersionless hierarchy}; the differential $Q$ becomes the differential $\delta$;  the bracket $\{-,-\}_C$ is identified with $\{-,-\}$.
The current observables $\oint_C  \J_{\mc G_{k,\alpha}}^{(1)}$ are identified with the Hamiltonian functions $\oint dz\ G_{k,\alpha}$. The proposition now follows from the same argument as in Subsection \ref{subsection:Dispersionless hierarchy}.
\end{proof}

This proposition is our interpretation of dispersionless integrable hierarchy via the classical BV master equation in BCOV theory.

\subsection{B-model interpretation}\label{Product geometry and push-forward}
In this subsection, we describe the B-model origin of $H$-valued BCOV theory using product of Calabi--Yau geometry and the recipe in Subsection \ref{sec:comp}. This will motivate our constructions in previous subsections and we will be sketchy.

Let $X, Y$ be two Calabi--Yau manifolds. We consider BCOV theory on $X\times Y$ and its compactification along $\pi_Y\colon X\times Y \to Y$. BCOV fields on $X\times Y$ are
$$
  \PV(X\times Y)\llb t\rrb=\PV(X)\otimes \PV(Y)\llb t\rrb.
$$
Let $\pa_X, \pa_Y$ denote the divergence operators on $\PV(X), \PV(Y)$ with respect to the corresponding Calab-Yau volume form. Let $\delta_X, \delta_Y$ be the $\delta$-distribution (with respect to the trace map)  of BCOV theory on $X, Y$ respectively. Then the BV kernel on $X\times Y$ is
$$
  K_0^{X\times Y}=\bracket{\pa_X+\pa_Y}\delta_X\otimes \delta_Y=K_0^X\otimes \delta_Y+ \delta_X\otimes K_0^Y.
$$
Here $K^X_0, K^Y_0$ are integral kernels (with respect to the trace map) of the divergence operators $\pa_X, \pa_Y$ on $X, Y$ respectively.

Let us choose a K\"{a}hler metric on $X$. Hodge theory implies the identification
$$
  H^\bullet(\PV(X)\llb t\rrb, \dbar_X+t\pa_X)\iso \H_X\llb t\rrb
$$
 where $\H_X$ are harmonic elements of $\PV(X)$.

  We consider the following contraction data where
 $$
 i\colon \ \ \H_X\otimes \PV(Y)\into \PV(X)\otimes \PV(Y)\llb t\rrb \quad \pi\colon\ \  \PV(X)\otimes \PV(Y)\llb t\rrb\to \H_X\otimes \PV(Y)\llb t\rrb
 $$
 are the harmonic embedding and projection. The parametrix is given by
 $$
    P=P_X\otimes \delta_Y+ G_X\otimes K_0^Y
 $$
 where
 $$
 P_X= \int_0^\infty (\dbar^*_X\pa_X\otimes 1)h_t^X, \quad G_X=\int_0^\infty (\dbar^*_X\otimes 1)h_t^X.
 $$
 Here $h_t^X$ is the heat kernel of $e^{-t[\dbar_X, \dbar^*_X]}$. We have
 $$
 K_0^{X\times Y}=\delta_{\H_X}\otimes K_0^Y+ \bbracket{Q_X, P}
 $$
 where $\delta_{\H_X}=h_{\infty}^X$ represents the integral kernel of the identity operator on $\H_X$.

Let $I_{X\times Y}$ be the BCOV action on $\PV(X)\otimes \PV(Y)\llb t\rrb$ as described in \cite{Si-Kevin}.  This is the same as in Definition \ref{vertex contribution} when $H=\C$ comes from the pure gravity. Since $P$ gives a homotopy between $K_0^{X\otimes Y}$ and $\delta_{\H_X}\otimes K_0^Y$, we define
$$
I_{\pi_Y}:=\sum_{\Gamma:\mr{trees}}W_\Gamma(I_{X\times Y}, P).
$$
By the same argument as in Theorem \ref{thm-push},  the functional $I_{\pi_Y}$ defines a solution of classical master equation for the BV data $\bracket{\H_X\otimes \PV(Y)\llb t\rrb , Q_Y, \delta_{\H_X}\otimes K_0^Y}$
$$
Q I_{\pi_Y}+{1\over 2}\{I_{\pi_Y},I_{\pi_Y}\}_{\mr{BV}}=0.
$$

Now we specialize to the case when $Y=\C$ (or $\C^\times$ or elliptic curve). $I_{\pi_\C}$ is a local functional on $\PV(\C)\otimes  \H_X\llb t\rrb$. It is built up from tree diagrams with $P=P_X\otimes \delta_\C+ G_X\otimes K_0^\C$ being the propagator. Observe that if we place the factor $G_X\otimes K_0^\C$ on the edge of some tree diagrams, it will produce a local functional containing holomorphic derivatives along $\C$.

Let $I^{D}_{\pi_\C}$ be the local functional collecting all those terms in $I_{\pi_\C}$ which do not contain any derivatives along $\C$. $I^{D}_{\pi_\C}$ can be viewed as a dispersionless limit of $I_{\pi_\C}$. Alternately, $I^{D}_{\pi_\C}$ is given by tree diagrams with $P_X\otimes \delta_\C$ being the propagator. Let $\abracket{-}_0$ be the B-model correlation function on $\H_X\llb t\rrb$, whose generating function is given by sum of tree diagrams for BCOV theory on $\PV(X)\llb t\rrb$. Then $I^{D}_{\pi_\C}$ defines the interaction of our $\H_X$-valued BCOV theory on $\C$ as in Definition \ref{vertex contribution}.

 By using a filtration by the number of derivatives, we can show that  $I^{D}_{\pi_\C}$ satisfies the same classical master equation as $I_{\pi_\C}$
$$
Q I^{D}_{\pi_\C}+{1\over 2}\{I^{D}_{\pi_\C},I^{D}_{\pi_\C}\}_{\mr{BV}}=0.
$$
By the argument in Subsection \ref{sec:current-integrable}, this gives the classical dispersionless integrable hierarchy for the B-model on $X$.

\subsection{Hierarchy from internal symmetries of BCOV theory}

The main claim of this subsection is that there are infinite-dimensional background fields in BCOV theory, which we can regard as coming from infinite-dimensional abelian symmetries. Noether's theorem implies that these should yield infinitely many mutually commuting current observables; this is the underlying reason why BCOV theory gives rise to integrable hierarchy. In this subsection, we flesh out this idea in some detail for the case of $H=\C$; the essential feature is all present in this case. As this is also purely motivational, we will be rather sketchy.

Let us consider BCOV theory on $Y$ of dimension $d$; we have $\E=\PV(Y)\llb t\rrb$ with the differential $Q=\bar {\pa }+t \pa $. As the propagator of BCOV theory pairs $\PV^{i,\bullet} (Y)$ with $\PV^{d-1-i,\bullet}(Y)$ for each $0 \leq i \leq d-1$, the space of \emph{dynamic fields} is the minimal cochain complex (respecting the interaction term) containing $\PV^{i,\bullet}$ for $i\leq d-1$. This is given by \[ \E_{ \mr{D} }= \bigoplus_{k+i \leq d-1} t^k \PV^{i,\bullet}(Y).\] The space of \emph{background fields}, which by definition are complementary in $\E$ to the dynamic fields, is given by \[\E_{ \mr{B} }= \bigoplus_{k\geq 0}  \left(t^k\PV^{d,\bullet} (Y) \to t^{k+1}\PV^{d-1,\bullet}(Y) \to \cdots \to t^{k+d}\PV^{0,\bullet}(Y)  \right).\]

\begin{prop}
The space of background fields $\E_{ \mr{B} }$ is quasi-isomorphic to $\C \llb t\rrb$.
\end{prop}

\begin{proof}
Consider the isomorphism $\PV(Y)\llp t \rrp \simeq \Omega(Y)\llp  t \rrp$ given by $t^k \alpha \mapsto t^{k+i-d}\alpha \vdash \Omega_Y$ for $\alpha \in \PV^{i,\bullet}(Y)$. This leads to a quasi-isomorphism
\[\xymatrix{
t^k\PV^{d,\bullet}\ar[r]^-{t\pa} \ar[d]_\simeq &  t^{k+1}\PV^{d-1,\bullet} \ar[r]^-{t\pa} \ar[d]^\simeq & \cdots \ar[r]^-{t\pa} &  t^{k+d}\PV^{0,\bullet} \ar[d]^\simeq  \\
t^k \Omega^{0,\bullet}  \ar[r]^-{\pa} & t^k \Omega^{1,\bullet} \ar[r]^-{\pa} & \cdots \ar[r]^-{\pa} & t^k\Omega^{d,\bullet}
}\]
Here $\pa$ in the above line is the divergence operator, and $\pa$ in the below line is the holomorphic de Rham differential. In other words, for each $k\geq 0$, we have a copy of de Rham complex $t^k \Omega_Y^\bullet$ with the usual degree of the complex. Thus we obtain $\E_{\mr{B}} \simeq\C \llb t \rrb $.
\end{proof}

To distinguish from other forms of a local functional, we write the BCOV interaction term as $I^{\mr{BCOV}}$; then $I^{\mr{BCOV}} \in \Ol(\mc E) = \Ol(\mc E_{\mr{D}}\oplus \mc E_{\mr{B}})$ can be understood as a map $I^{\mr{BCOV}}\colon \mc E_{\mr{B}}\to \Ol(\mc E_{\mr{D}})$. In fact, this can be regarded as an $L_\infty$-map $\mc E_{\mr{B}}[-1]\to \Ol(\mc E_{\mr{D}})[-1]$ where the cochain complex $\mc E_{\mr{B}}[-1]$ is regarded as an $L_\infty$-algebra, thanks to the above proposition. Then, $\mc E_{\mr{B}}[-1]$ represents the infinite-dimensional abelian symmetry algebra and its image should define currents which are mutually commuting.

The case of interest for our paper is when $Y=\C^\times $. Recall that we write an arbitrary element $\mu \in \PV(\C^\times)\llb t \rrb$ as $\mu =\sum_{k\geq 0} (\lambda _k t^k + \rho _k \pa_z  t^k)$ where $\lambda_k\in \Omega^{0,\bullet}(\C^\times )$ and $\rho_k\in \Omega^{0,\bullet}(\C^\times )$. According to the above decomposition, we have
\[\mc E_{\mr{D}}=\PV^{0,\bullet}(\C^\times )\qquad \text{and}\qquad  \mc E_{\mr{B}} = \oplus_{k\geq 0} \left(t^k \PV^{1,\bullet}(\C^\times) \stackrel{t\pa }{\longrightarrow} t^{k+1}\PV^{0,\bullet}(\C^\times) \right)\]
where $\lambda_0\in \mc E_D= \PV^{0,\bullet}(\C^\times)$ and $\mc E_B$ consists of $(\rho_k,\lambda_{k+1})$ for each power $k$ of $t$. The above proposition says that up to cohomology we can represent $(\rho_k,\lambda_{k+1})$ by requiring $\rho_k=c_k$ to be constant and $\lambda_{k+1}=0$ for $k\geq 0$. One should think of this as the origin of the stationary sector $\mc E_S$ we have considered.

The BCOV action functional $I^{\mr{BCOV}} = \Tr \langle e^\mu\rangle_0$ is linear in $\rho_k$'s for a type reason. If we restrict to $\mc E_S$ where $\rho_k=c_k$ are constant and $\lambda_{k+1}=0$ for $k\geq 0$, and  define $I^{\mr{BCOV}}_S := I^{\mr{BCOV}}|_{\mc E_S}$, then one has
\[I^{\mr{BCOV}}_S = \Tr \abracket{ \sum_{k\geq 0}  \rho_k \pa_z t^k  e^{\lambda_0}}_0= \sum_{k\geq 0} \int_{\C^\times } \left(\abracket{ (\rho_k \pa_z t^k e^{\lambda_0})}_0 \vdash dz \right)\wedge dz = \sum_{k\geq 0 } \int_{\C^\times }dz\ \rho_k \abracket{t^k e^{\lambda_0}}_0   .\]
Now let us interpret this as a map $\mc E_{\mr{B}}[-1]\to \Ol(\mc E_D)[-1]$. That is, the image of $\rho_k \pa_z t^k\in \mc E_B[-1]$ with $\rho_k=c_k$ leads to the current observable $\lambda_0 \mapsto  \oint_C dz \abracket{t^k e^{\lambda_0} }_0  $. Under our identification for the interpretation of dispersionless integrable hierarchy in terms of BCOV theory, the current exactly becomes $G_k$. The situation is summarized in the  table:

\begin{center}
\begin{tabular}{|c|c|}
\hline
Dispersionless Integrable Hierarchy & BCOV Theory\\
\hline
 $b_k,\eta_k $  & $\lambda_k,\rho_k  $  \\
 \hline
 \multirow{2}{*}{ $I = \sum_{k } \oint dz \ \eta_k  \abracket{   t^k e^{ \sum b_k t^k } }_0 $} & $I^{\mr{BCOV}}= \int_{\C^{\times}}\J_{\mc L}^{(2)}=  \sum_k \int_{\C^\times } dz \ \rho_k \abracket{ t^k e^{\sum_m \lambda_m t^m }  }_0 $ \\
 &  $\oint_C \J_{\mc L}^{(1)}=\sum_k \oint_{C} dz \ \rho_k \abracket{ t^k e^{\sum_m \lambda_m t^m }  }_0 $   \\
\hline
 $I_S = \sum_{k} \oint dz\ c_k \abracket{  t^k e^{b_0} }_0  $  &  $ I^{\mr{BCOV}}_S = \sum_k \int_{\C^{\times}}\J_{\mc G_k}^{(2)}=   \sum_{k}  \int_{\C^\times } dz\  c_k \abracket{t^k e^{\lambda_0}}_0$  \\
 $  G_{k } = \abracket{   t^k e^{b_0} }_0  $ & $ \J_{\mc G_k}^{(1)}=dz\ \abracket{t^k e^{\lambda_0}}_0$  \\
\hline
\end{tabular}
\end{center}

\end{document}